\documentclass[journal]{IEEEtran}
\IEEEoverridecommandlockouts

\usepackage[T1]{fontenc}
\usepackage[latin9]{inputenc}
\usepackage[english]{babel}
\usepackage{color}

\ifCLASSINFOpdf
  \usepackage[pdftex]{graphicx}					
  \DeclareGraphicsExtensions{.pdf,.jpeg,.png,.eps}
\else
  \usepackage[dvips]{graphicx}
  \DeclareGraphicsExtensions{.eps}
\fi
\usepackage[caption=false]{subfig}
\graphicspath{{Figures/}}			
\usepackage{ifpdf}
\ifpdf
  \usepackage{epstopdf}
\fi

\usepackage{amsmath}
\usepackage{mathtools}
\usepackage{amssymb}
\usepackage{amsthm}
\usepackage{bbm}

\usepackage{enumerate}
\usepackage{enumitem}
\usepackage{array}
\usepackage{multirow}

\usepackage{cite}
\usepackage{cleveref}

\usepackage{algorithm}
\usepackage{algorithmic}

\crefname{app}{Appendix}{Appendices}
\newtheorem{theorem}{Theorem}
\newtheorem{cor}{Corollary}

\newtheorem{prop}{Proposition}

\newtheorem{defn}{Definition}
\newtheorem{exam}{Example}

\crefname{cor}{Corollary}{Corollary}
\crefname{prop}{Proposition}{Proposition}
\crefname{lemma}{Lemma}{Lemma}

\newcommand{\bs}{\boldsymbol}
\newcommand{\bb}{\mathbb}
\newcommand{\eye}{\bs{I}}
\newcommand{\zero}{\bs{0}}

\newcommand{\lb}{\left(}
\newcommand{\rb}{\right)}

\newcommand{\lc}{\left\{}
\newcommand{\rc}{\right\}}
\newcommand{\ld}{\left.}

\newcommand{\lv}{\left\vert}
\newcommand{\rv}{\right\vert}
\newcommand{\lV}{\left\Vert}
\newcommand{\rV}{\right\Vert}
\newcommand{\LRV}[1]{{\left\vert\kern-0.25ex\left\vert\kern-0.25ex\left\vert #1 \right\vert\kern-0.25ex\right\vert\kern-0.25ex\right\vert}}

\newcommand{\nth}{^\mathrm{th}}
\newcommand{\tran}{^{\mathsf{T}}}

\newcommand{\rank}[1]{\mathrm{Rank}\lc#1\rc}

\newcommand{\matA}{\bs{A}}

\newcommand{\calCS}{\mathcal{CS}}
\newcommand{\bbC}{\bb{C}}
\newcommand{\matD}{\bs{D}}

\newcommand{\calH}{\mathcal{H}}
\newcommand{\matH}{\bs{H}}

\newcommand{\matR}{\bs{R}}
\newcommand{\bbR}{\bb{R}}

\newcommand{\calS}{\mathcal{S}}

\newcommand{\matU}{\bs{U}}
\newcommand{\matV}{\bs{V}}

\newcommand{\matW}{\bs{W}}

\newcommand{\vech}{\bs{h}}

\newcommand{\vecu}{\bs{u}}

\newcommand{\vecx}{\bs{x}}
\newcommand{\vecy}{\bs{y}}
\newcommand{\vecz}{\bs{z}}

\newcommand{\vecalpha}{\bs{\alpha}}
\newcommand{\vecbeta}{\bs{\beta}}

\newcommand{\vecgamma}{\bs{\gamma}}

\newcommand{\matPsi}{\bs{\Psi}}

\usepackage{flushend}

\makeatletter 
\pretocmd\@bibitem{\color{black}\csname keycolor#1\endcsname}{}{\fail}
\newcommand\citecolor[1]{\@namedef{keycolor#1}{\color{blue}}}
\makeatother

\begin{document}

\title{Controllability of Linear Dynamical Systems Under Input Sparsity Constraints}
\author{Geethu Joseph and Chandra R. Murthy {\it Senior Member, IEEE}
\thanks{The authors are with the Dept.\ of ECE at IISc, Bangalore, India, Emails:\{geethu, cmurthy\}@iisc.ac.in.}
\thanks{The work of G. Joseph was supported in part by the Intel India PhD fellowship, and the work of C. R. Murthy was supported in part by the MeitY Young Faculty Research Fellowship.}
}
\maketitle
\begin{abstract}
In this work, we consider the controllability of a discrete-time linear dynamical system with \emph{sparse} control inputs. 
Sparsity constraints on the input arises naturally in networked systems, where activating each input variable adds to the cost of control. 
We derive algebraic necessary and sufficient conditions for ensuring controllability of a system with an arbitrary transfer matrix. The derived conditions can be verified in polynomial time complexity, unlike the more traditional Kalman-type rank tests. 
Further, we characterize the minimum number of input vectors required to satisfy the derived conditions for controllability. Finally, we present a generalized Kalman decomposition-like procedure that separates the state-space into subspaces corresponding to sparse-controllable and sparse-uncontrollable parts. These results form a theoretical basis for designing networked linear control systems with sparse inputs.
\end{abstract}

\begin{IEEEkeywords}
Linear dynamical systems, sparsity, controllability, Kalman rank test, PBH test, switched linear systems.
\end{IEEEkeywords}

\section{Introduction}
\label{sec:intro}
Networked control systems 
have attracted intense research attention from both academia and industry over the past  decades~\cite{Pasqualetti_Controllability_2014, Tatikonda_Control_2000,Liu_Control_2016,Chanekar_Optimal_2017,Nozari_Time_2017}. In such a system, 
the notion of controllability refers to the ability to drive the system from an arbitrary initial state to a desired final state in a finite amount of time. Complete characterization of controllability of linear dynamical systems using unconstrained inputs have pure algebraic rank-based forms, and are rather easily verifiable~\cite{Kalman_General_1959,
Hautus_Stabilization_1970}. 
These conditions involve verification of the rank conditions of suitably defined matrices. 
However, in applications involving networked control systems, it is often necessary to select a small subset of the available sensors or actuators at each time instant, due to communication bandwidth, cost, or energy constraints~\cite{Nagahara_Sparse_2011,Li_Sparse_2016}. Further, it is often desirable to select a different subset of nodes at each time instant to improve the network lifetime~\cite{Jadbabaie_Deterministic_2018}.  
{For example:
\begin{itemize}
\item In an energy-constrained network, energy-aware scheduling of actuators can help to extend the battery life of the nodes~\cite{Jadbabaie_Deterministic_2018}. While choosing a small subset of nodes at each time instant helps in reducing the control overhead, repeatedly using the same set of nodes over time drains the batteries of the selected nodes. Hence, it is desirable to choose a different subset of nodes at each time instant to improve the network lifetime. 
\item In a system where the controller and plant communicate over a network, the control signals are required to meet the bandwidth constraints imposed by the links over which they are exchanged~\cite{Nagahara_Sparse_2011,Li_Sparse_2016}.  Using a non-sparse control input requires higher bandwidth, and restricting the control signals to a fixed support may severely limit the set of admissible inputs to the system. On the other hand, using different supports provides much greater flexibility without significantly increasing the communication requirements. Therefore, this strategy combines the benefits of the two approaches.
\item The opinion dynamics in a social network is often modeled using a linear opinion propagation framework~\cite{De_Learning_2014,Wendt_Control_2019}. Here, the state of the system is denoted by a vector containing the opinion of each individual in the network, and the transition matrix is determined by the network topology. Further, it is assumed that an agent  desires to drive the network opinion to a particular state by influencing only a few people on the network. For example, a company may distribute free samples of its products to  some members of the network, under a budget constraint on the number of free samples distributed. Also, for better marketing, the company may want to give the free samples to different members over time, instead of giving samples to the same set of people. Here also, the support of the sparse control signal varies over time.
\item In an airplane environmental control system, the air quality in the cabin is maintained by operating several valves onboard the aircraft. Here, the state of the system is represented by a vector containing the value of a desired parameter (for example, the carbon dioxide level) across the airplane, and opening or closing a valve is tantamount to applying a specific vector-valued input. In this case also, it is desirable to maintain the air quality by operating as few valves as possible at each time instant, and changing the support is helpful to avoid wearing out a fixed set of valves and possibly to achieve faster control. 
\end{itemize}}
Now, when the number of actuators or input variables that can be activated at each time instant is limited, the system may become uncontrollable because all the feasible control signals are restricted to lie in the union of low-dimensional subspaces.  The controllability of linear dynamical systems under sparse input constraints is the focus of this paper.

\subsection{ Related Literature}
We first discuss the relationship between the problem considered in this paper and the existing literature in control theory and sparse signal processing.

\subsubsection{Time-varying actuator scheduling problem}

This problem focuses on finding a schedule for sparse actuator control, such that the system is sparse-controllable~\cite{Chanekar_Optimal_2017,Nozari_Time_2017,
Jadbabaie_Deterministic_2018}. These works rely on a well known condition for controllability, namely, an extended version of the Kalman rank test. This test depends on the rank of the so-called Gramian matrix of the sparsity-constrained system.\footnote{Refer to \cite[Section II.B]{Jadbabaie_Deterministic_2018} for details.} However, finding sequence of control inputs that satisfy the rank condition on the Gramian matrix is a combinatorial problem, and it is known to be NP-hard~\cite{Olshevsky_Minimal_2014,Tzoumas_Minimal_2016}. Moreover, it has been recently shown that the relatively simpler problem of finding a sparse set of actuators to guarantee reachability of a particular  state is hard to approximate, even when a solution is known to exist~\cite{Jadbabaie_Minimal_2018}. Hence, different quantitative measures of controllability based on the Gramian matrix have been considered: smallest eigenvalue, the trace of the inverse, inverse of the trace, the determinant, maximum entry in the diagonal, etc.~\cite{Jadbabaie_Deterministic_2018}. Based on these metrics, several algorithms and related guarantees are available in the literature~\cite{Pasqualetti_Controllability_2014,Chanekar_Optimal_2017,Nozari_Time_2017}. However, none of the above mentioned references directly address the fundamental question of whether or not the system can be controlled by sparse inputs. Further, direct extension of the Kalman rank test leads to a combinatorial problem that is computationally infeasible to solve in practice. Thus, the goal of our paper is  to study the controllability of a linear dynamical system under sparsity constraints without directly relying on Gramian matrix. We are not interested in finding the optimal actuator selection; rather we deal with the more basic problem of deriving conditions for the existence of a selection that drives the system from any initial state to any final state. 

\subsubsection{Minimal input selection problem}The minimal input selection involves selecting a small set of input variables so that the system is controllable using the selected set~\cite{Liu_Minimal_2017, Tzoumas_Minimal_2016,Olshevsky_Minimal_2014}. This problem is a special case of our sparse input problem because of the extra constraint that the support of the control input remains unchanged for all time instants. Moreover, the controllability conditions for the minimal input selection problem can be easily be derived from the classical controllability results for the unconstrained system. We discuss and contrast the two cases in detail in \Cref{sec:mmv_sparse}.

\subsubsection{Design of sparse control inputs}
Some recent works connecting compressive sensing and control theory focus on the design of control inputs~\cite{Charles_Short_2014, Sefati_Linear_2015, Kafashan_Relating_2016}. 
They propose algorithms for the recovery (design) of sparse control inputs based on the observations, and derive conditions under which the input can be uniquely recovered using a limited number of observations~\cite{Charles_Short_2014, Sefati_Linear_2015, Kafashan_Relating_2016}. These problems do not deal with controllability related issues, rather assume the existence of sparse control inputs and initial state for reaching a given final state.

\subsubsection{Observability under sparsity constraints} Due to the recent advances in sparse signal processing and compressed sensing, researchers have recently started looking at the observability of linear systems with a sparse initial state~\cite{Dai_obsrvability_2013,Wakin_observability_2010,Joseph_Observability_2018, Joseph_Observability_2019_TSP}. For a system with unconstrained inputs, observability and controllability are dual problems and do not require separate analysis. However, our problem assumes a general initial state and sparse control inputs, whereas~\cite{Joseph_Observability_2018, Dai_obsrvability_2013,Wakin_observability_2010, Joseph_Observability_2019_TSP} consider a sparse initial state and 
known control inputs. Therefore, the problems have different sparsity pattern models, and consequently require separate analysis.
\subsubsection{Sparse signal recovery guarantees}The sparse controllability problem studies the conditions that ensure the existence of sparse control inputs to drive a linear system from any given state to any other state. Moreover, it is not required that the solution be unique. In contrast, the focus of traditional sparse signal processing studies is on developing algorithms and guarantees for the cases where the linear system is already known to admit a sparse solution~\cite{Foucart_math_2013,Donoho_Compressed_2006,Candes_Robust_2006,Baraniuk_Compressive_2007}. Also, the structure of the effective measurement matrix that arises in the context of linear dynamical systems is different from the type of random measurement matrices that are usually considered in the compressed sensing literature. 

The problem of controllability using sparse inputs is completely different in flavor compared to the existing work in control theory. Also, the solution to the problem is cannot be obtained using any of the available tools from sparse signal processing.

\subsection{ Our Contributions}
In this paper, we answer the following key questions:
\begin{enumerate}[label=Q\arabic*]
\item What are necessary and sufficient conditions for ensuring controllability under sparse input constraints? Can we devise a computationally simple  test for controllability? \label{q:Q1}
\item If a system is controllable using sparse inputs, how many control input vectors needed to drive the system from a given initial state to an arbitrary final state?\label{q:Q2}
\item If the system is not controllable using sparse inputs, what part of the state space is reachable using sparse inputs? 
\label{q:Q3}
\end{enumerate} 
Answering the above questions requires a fresh look at controllability, and we start by deriving a Popov-Belevitch-Hautus (PBH)-like test~\cite{Hautus_Stabilization_1970}, which, unlike the Gramian matrix-based tests, allows one to check for sparse-controllability of a system without solving a combinatorial problem. Our specific contributions are as follows:
\begin{enumerate}
\item We establish a set of the necessary and sufficient conditions for the controllability of a linear system under sparse inputs in \Cref{sec:control}. Using these conditions, we present a simple procedure to check the controllability of any system using sparse inputs. 
\item We upper and lower bound   the minimum number of input vectors that can steer the system from any given initial state to any desired final state in \Cref{sec:length}. 
We show that the upper bound is no more than the length of the state vector, which is also an upper bound for the minimum number of input vectors for an unconstrained system.
\item We present a procedure to convert a representation of any linear dynamical system into a \emph{standard form} in \Cref{sec:decomposition}. The standard form separates the state-space into uncontrollable, sparse-uncontrollable  and sparse-controllable components. 
\end{enumerate}
In a nutshell, this paper presents new results on the controllability of linear dynamical systems under sparsity constraints on the input. We also note that the classical results for the unconstrained system can be recovered as a special case of our results, by relaxing the sparsity constraint. 
 
\textbf{Notation:} In the sequel, 
we use $\lv \cdot \rv$ to denote the cardinality of a set and $\lV  \cdot \rV_0$ to denote the $\ell_0$ norm of a vector. For any positive integer $a$, $[a]$ denotes the set $\lc1,2,\ldots,a\rc$. The symbols $\eye$ and $\zero$ represent the identity matrix and the all zero matrix (or vector), respectively. The notation $\matA_{i}$  denotes the $i\nth$ column of the matrix $\matA$, and $\matA_{\calS}$ represents the submatrix of $\matA$ formed by the columns indexed by the set $\calS$. Also, $\calCS\lc\cdot\rc$, $\rank{\cdot}$ and $\lb\cdot\rb\tran$ represent the column space, rank and transpose of a matrix, respectively.

\section{System Model}
We consider the  discrete-time linear dynamical system whose state at time $k$, denoted by $\vecx_k \in \bbR^N$, evolves as
\begin{equation}
\vecx_k = \matD\vecx_{k-1}+\matH\vech_{k},\label{eq:sys}
\end{equation}
where the transfer matrix $\matD\in\bbR^{N\times N}$
and input matrix $\matH\in\bbR^{N\times L}$. Here, the input vectors $\vech_k \in \bbR^L$ are assumed to be sparse, i.e., $\lV\vech_k\rV_0\leq s$, for all values of $k$. We denote the rank of the matrices $\matD$ and $\matH$ using $R_{\matD}$ and $R_{\matH}$, respectively.

We formally define the notion of controllability using sparse inputs as follows:
\begin{defn}[Sparse-controllability]\label{def:sparsecontrol}
The system in \eqref{eq:sys} is said to be  $s$-sparse-controllable if, for any initial state $\vecx_0 = \vecx_{\text{init}}$ and any final state $\vecx_{\text{final}}$, there exists inputs $\lc \vech_k\rc_{k=1}^K$ such that $\lV\vech_k\rV_0\leq s$, which steers the system from the state $\vecx_0 = \vecx_{\text{init}}$ to $\vecx_K = \vecx_{\text{final}}$ for some finite~$K$.  
\end{defn}
Next, to characterize the sparse-controllability of the system, we consider the following equivalent system of equations:
\begin{equation}\label{eq:sys_concat}
\vecx_{K}-\matD^K\vecx_0 = \tilde{\matH}_{(K)}\vech_{(K)},
\end{equation}
where we define the matrices as follows:
\begin{align}
\tilde{\matH}_{(K)} &= \begin{bmatrix}
 \matD^{K-1}\matH & \matD^{K-2}\matH &\ldots \matH
\end{bmatrix}\in\bbR^{N\times KL}\label{eq:control_mat_defn}\\
 \vech_{(K)}&=\begin{bmatrix}
\vech_1\tran&\vech_2\tran&\ldots&\vech_{K}\tran
\end{bmatrix}\tran\in\bbR^{KL}.
\end{align}
Note that $\vech_{(K)}$ is a \emph{piecewise sparse vector} formed by concatenating $K$  vectors, each with sparsity at most~$s$. 

\section{Necessary and Sufficient Conditions for Sparse-Controllability} \label{sec:control}
This section addresses question \ref{q:Q1} in \Cref{sec:intro}. 
Now,  it is known that the system is sparse-controllable if, for some finite $K$, there exists index sets $\lc\calS_i\rc_{i=1}^K$,  $\calS_i\subseteq\lc 1,2,\ldots,L\rc$, $\lv\calS_i\rv=s$, for $i=1,2,\ldots,K$, such that the following submatrix of $\tilde{\matH}_{(K)}$ has rank $N$:
\begin{equation}\label{eq:submatrix_form}
\begin{bmatrix}
\matD^{K-1}\matH_{\calS_1} & \matD^{K-2}\matH_{\calS_2} & \ldots & \matH_{\calS_K}
\end{bmatrix}\in\bbR^{N\times Ks}.
\end{equation}
In the sequel, we refer this condition to as the \emph{Kalman-type rank test}. Note that the first $(K-1)N$ columns of $\tilde{\matH}_{(K)}$ belong to $\calCS\lc\matD\rc$. Hence, to satisfy the Kalman-type rank test, $\calS_{K}$ should be such that $\calCS\lc\matH_{\calS_{K}}\rc
$ should contain the left null space of $\matD$.  {Thus, a necessary condition for sparse-controllability is the existence of an index set $\calS$ with $s$ entries such that  $\rank{\begin{bmatrix}
\matD & \matH_{\calS}
\end{bmatrix}} = N$,
which is possible only if $s\geq N-R_{\matD}$. Further, a system can be sparse-controllable only if it is controllable using unconstrained inputs. Therefore, for sparse-controllability, it is necessary that the system is controllable and $s\geq N-R_{\matD}$. In fact, these two conditions are not only necessary but also sufficient, as we show in the following theorem:}


{\begin{theorem}\label{thm:necessary_sufficient}
The  system in \eqref{eq:sys} is $s$-sparse-controllable if and only if $\rank{\begin{bmatrix}
\lambda\eye-\matD & \matH
\end{bmatrix}}=N\leq s+R_{\matD}
$
for all $\lambda\in\bbC$.
\end{theorem}}
\begin{proof}
See \Cref{app:necessary_sufficient}.
\end{proof}
{Note that  there are two separate conditions here: one, a condition on the rank of the matrix $\begin{bmatrix}
\lambda\eye-\matD & \matH
\end{bmatrix}\in\bbR^{N+L}$, which we refer to as \emph{the rank condition of \Cref{thm:necessary_sufficient}}; and two, a lower bound on the sparsity $s$, which we refer to as \emph{the inequality condition of \Cref{thm:necessary_sufficient}}.
The rank condition is same as the classical  PBH test~\cite{Hautus_Stabilization_1970} which is \emph{independent of the sparsity level $s$}, while the inequality condition is \emph{independent of the input matrix $\matH$}.}
We make the following further remarks: 

\begin{itemize}
\item 
 A reversible system, i.e., a system with an invertible state transition matrix $\matD$, is $s$-sparse-controllable for any $0 < s\leq L$ if and only if it is controllable. Similarly, when $L=1$, the notion of sparse-controllability and controllability are the same, and hence \Cref{thm:necessary_sufficient} reduces to the PBH test.
\item If the system defined by the  matrix pair $(\matD,\matH_{\calS})$ is controllable for some index set $\calS$ with $s$ entries, the system is $s$-sparse-controllable. In particular, a controllable system with $R_{\matH}\leq s$ is $s$-sparse-controllable. 

{\item 
 The system given by \eqref{eq:sys} is controllable using inputs that are $s$-sparse under a basis $\matPsi\in\bbR^{L\times L}$ if and only if the system is controllable using inputs that are $s$-sparse under the canonical basis. This  follows by replacing $\matH$ with $\matH\matPsi$ in \Cref{thm:necessary_sufficient}, and noting that for any $\lambda\in\bbC$, 
\begin{equation}
\rank{\begin{bmatrix}
\lambda\eye-\matD & \matH\matPsi
\end{bmatrix}}  =\rank{\begin{bmatrix}
\lambda\eye-\matD & \matH
\end{bmatrix}}. \nonumber
\end{equation}
\item 
 The verification of sparse-controllability has the same complexity as the classical PBH test. This is because, we only need to additionally check the inequality in \Cref{thm:necessary_sufficient}, and  $R_{\matD}$ is already known from the PBH test. 
Thus, 
\Cref{thm:necessary_sufficient} allows us to verify the controllability of any discrete system in polynomial complexity in $N$,
independent of the sparsity $s$. On the other hand, to verify the Kalman-type rank test, we need to perform $\binom{L}{s}^N$ rank computations.
Further, since the Kalman-type rank test involves powers of $\matD$, numerical stability also needs to be considered. 

}
\end{itemize}
{ 
\subsection{Output Controllability}
We consider the linear dynamical system described by \eqref{eq:sys} and the following output relation:
\begin{equation}
\vecy_k = \matA\vecx_k.\label{eq:sys_2}
\end{equation}
where the output matrix $\matA\in\bbR^{m\times N}$ with $m<N$. 
Similar to \Cref{def:sparsecontrol}, we define the notion of \emph{output $s$-sparse-controllability} as the existence of an $s$-sparse sequence of inputs which steers the system from initial state $\vecx_0$ to a final output $\vecy_K$, for some finite $K$.  Now, to characterize the output sparse-controllability, we consider the following equivalent system of equations:
\begin{equation}\label{eq:sys_concat_out}
\vecy_{K}-\matA\matD^K\vecx_0 = \matA\tilde{\matH}_{(K)}\vech_{(K)}.
\end{equation}

In  \cite{Westphal_Handbook_2012}, a Kalman test for output controllability of an \emph{unconstrained} system is derived, which states that the system given by \eqref{eq:sys} and \eqref{eq:sys_2} is output controllable if and only if the matrix $\matA\tilde{\matH}_{(K)}$ has full row rank for some finite $K$. However, a direct extension of this result to the case of output \emph{sparse-}controllability leads to a computationally expensive combinatorial test as follows. The system is output controllable if and only if, for some finite $K$, there exists a submatrix of $\matA\tilde{\matH}_{(K)}$ with rank $m$ of the form
\begin{equation*}
\matA\begin{bmatrix}
\matD^{K-1}\matH_{\calS_1} & \matD^{K-2}\matH_{\calS_2} & \ldots & \matH_{\calS_K}
\end{bmatrix}\in\bbR^{m\times Ks},
\end{equation*}
 such that the index set $\calS_i\subseteq\lc 1,2,\ldots,L\rc$ and $\lv\calS_i\rv=s$, for $i=1,2,\ldots,K$.
Hence, we first present the following PBH test-type result for output (unconstrained) controllability:
\begin{prop}
\label{prop:classical_out}
For an unconstrained system given by \eqref{eq:sys} and \eqref{eq:sys_2}, the system is output controllable only if,  for all $\lambda\in\bbC$, 
the rank of $\matA\begin{bmatrix}
\lambda\eye - \matD & \matH
\end{bmatrix}\in\bbR^{m\times (N+L)}$ is~$m$.
\end{prop}
\begin{proof}
Our proof is by contradiction. Suppose that, for some $\lambda\in\bbC$, the matrix $\matA\begin{bmatrix}
\lambda\eye - \matD & \matH
\end{bmatrix}$ does not have full row rank. Then, there exists a $\zero \neq \vecz\in\bbC^m$ such that 
\begin{equation}
\vecz\tran\matA\matD = \lambda\vecz\tran\matA
\text{ and }\vecz\tran\matA\matH = \zero,
\end{equation}
which implies $\vecz\tran\matA\tilde{\matH}_{(K)} = \zero$ for all $K$.
Hence, the Kalman test is violated, and the system is not output controllable.
\end{proof}
\begin{exam}
Let $m=3$, $N=5$ and $L=3$, and suppose the system  given by \eqref{eq:sys} and \eqref{eq:sys_2} is defined by the following matrices:
\begin{align}
\matD &= \begin{bmatrix}
	1 & 2 & 4 & 5 & 9\\
    7 & 2 & 3 & 1 & 7\\
    0 & 0 & 1 & 2 & 5\\
    0 & 0 & 3 & 4 & 7\\
    0 & 0 & 1 & 6 & 9
\end{bmatrix}, 
\matH = \begin{bmatrix}
	 1\\
     2\\
     0\\
     0\\
     0
\end{bmatrix}, \\
\matA &= \begin{bmatrix}
0   & 0.019 &   -0.14 &    0.02 &    0.99\\
0   & -0.08 &    0.24 &   0.97 &   0.018\\
1   &    0   &     0   &      0 &        0
\end{bmatrix}.
\end{align}
It can be verified that the system fails the Kalman test, as $\rank{\matA\tilde{\matH}_{(K)}}<m$ for all $K$. However, the condition of \Cref{prop:classical_out} is satisfied. Thus, the condition is necessary but not sufficient for output controllability.
\end{exam}

Our extension of \Cref{thm:necessary_sufficient} to  output sparse-controllability is as follows:

\begin{cor}\label{cor:necessary_sufficient_out}
The system given by \eqref{eq:sys} and \eqref{eq:sys_2} is output $s$-sparse-controllable
only if $s\geq m-\rank{\matA\matD}$, and for all $\lambda\in\bbC$, the rank of $\matA\begin{bmatrix}
\lambda\eye-\matD & \matH
\end{bmatrix}\in\bbR^{m\times (N+L)}$ is~$m$. 
\end{cor}
\begin{proof}
The proof is similar to that of \Cref{thm:necessary_sufficient} in \Cref{app:necessary_sufficient}. We replace $\vecz$ in the last part of the proof with $\matA\vecz$ to show the necessity of the above conditions.
\end{proof}

\Cref{cor:necessary_sufficient_out} is the same as \Cref{thm:necessary_sufficient}, except for a pre-multiplication with $\matA$. We make the following observations:
\begin{itemize}
\item We note that $\rank{\matA\matH^*}\leq \rank{\matA}$ for any matrix $\matH^*$. Hence, if  $\rank{\matA}<m$, the Kalman test 
fails and the system is not output sparse-controllable. 
\item 
 Suppose $\rank{\matA}=m$ for an $s$-sparse-controllable system. Invoking Sylvester's rank inequality~\cite{Hohn_Elementary_2013}, we get
\begin{multline}
m=\rank{\matA}+\rank{\matH^*}-N \\\leq \rank{\matA\matH^*}\leq \rank{\matA}=m,
\end{multline}
where $\matH^*\in\bbR^{N\times Ks}$ is the submatrix of $\tilde{\matH}_{(K)}$ that satisfies the Kalman test for state sparse-controllability, for some finite $K$. Hence, the system is output $s$-sparse-controllable. Therefore, when $\rank{\matA}=m$, the conditions in \Cref{cor:necessary_sufficient_out} are less restrictive than those in \Cref{thm:necessary_sufficient}, as the output dimension $m\leq N$.
\end{itemize}
\begin{exam}
Let $m=2$, $N=3$ and $L=2$, and suppose the system  given by \eqref{eq:sys} and \eqref{eq:sys_2} is defined by the following matrices:
\begin{equation}
\matD=\begin{bmatrix}
1&0&0\\
0&0&0\\
0&0&0
\end{bmatrix},\hspace{0.3cm} \matH=\begin{bmatrix}
1  &   1 \\
1  &   0\\
0  &  1
\end{bmatrix} \text{ and } 
\matA = \begin{bmatrix}
1   & 0 &  0\\
0   & 1 &    0 
\end{bmatrix}.
\end{equation}
It can be verified that the system is not $1-$sparse-controllable, but the system is output $1-$sparse-controllable.
\end{exam}

}
\subsection{Inputs with Common Support}\label{sec:mmv_sparse}
We recall the minimal input selection problem discussed in \Cref{sec:intro}. For such a problem, the system is controlled using sparse inputs with a common support, i.e., when the indices of the nonzero entries of all the inputs coincide. In this case, the effective system has the transfer matrix-input matrix pair as $(\matD,\matH_{\calS})$ for some index set $\calS$ such that $\lv\calS\rv=s$. Hence, the controllability conditions are given as follows:
\begin{enumerate}[label=(\roman*)]
\item For some finite $K$, there exists a $N\times Ks$ submatrix $\begin{bmatrix}
\matD^{K-1}\matH_{\calS} \  \matD^{K-2}\matH_{\calS} \  \ldots \  \matH_{\calS}
\end{bmatrix}$ of $\tilde{\matH}_{(K)}$ with rank $N$, 
 where $\calS\subseteq[L]$ and $\lv\calS\rv=s$.

\item For all $\lambda\in\bbC$, rank of $\begin{bmatrix}
\lambda\eye - \matD & \matH_{\calS}
\end{bmatrix}\in\bbR^{N\times (N+s)}$ is $N$, for some $\calS\subseteq\lc 1,2,\ldots,L\rc$ such that $\lv\calS\rv=s$. \label{con:static_support}
\end{enumerate}
{Clearly, \ref{con:static_support} above implies the two conditions of \Cref{thm:necessary_sufficient}. Therefore,}  the above conditions are more stringent than those in \Cref{thm:necessary_sufficient}, which is expected due to the additional requirement of using a common support. Thus, a system with sparse inputs with time-varying support offers greater flexibility and control, and incurs a similar communication cost,\footnote{The communication cost remains of order $s$, since the support can be conveyed using $s\log(L)$ bits.} compared to a system that uses sparse inputs with a common support. 

{
From the PBH-type condition, $s$-sparse-controllability with a common support holds only if 
\begin{equation}\label{eq:rank_bound_static}
\min\lc R_{\matH},s\rc\geq g_{\matD}\geq N-R_{\matD},
\end{equation}
where $g_{\matD}$ is the largest geometric multiplicity of an eigenvalue of $\matD$.
\subsection{Illustrative Examples}\label{sec:exam_1}
We first give an example to demonstrate that a controllable system which does not satisfy the inequality condition of \Cref{thm:necessary_sufficient} is not sparse-controllable.
\begin{exam}
Consider a linear system with $N=3$, $L=2$,
\begin{equation}
\matD=\begin{bmatrix}
1&0&0\\
0&0&0\\
0&0&0
\end{bmatrix} \text{, and } \matH=\begin{bmatrix}
1  &   1 \\
1  &   0\\
0  &  1
\end{bmatrix}.
\end{equation} 
Using the PBH test, it is easy to see that the system is controllable. However, the system does not satisfy the conditions of \Cref{thm:necessary_sufficient}. 

We verify that the system is not $1-$sparse-controllable using the initial state $\vecx_0=\zero$ and final state $\vecx_f=\begin{bmatrix}
1& 1& 1
\end{bmatrix}\tran$. From \eqref{eq:sys_concat}, we have,
\begin{equation}
\begin{bmatrix}
1\\
1\\
1
\end{bmatrix} = \sum_{k=1}^K\matD^{K-k}\matH\vech_k = \begin{bmatrix}
\sum_{k=1}^K\vech_k[1]+\vech_k[2]\\
\vech_K[1]\\
\vech_K[2]
\end{bmatrix}.
\end{equation}
Since $\vech_K$ is $1-$sparse, the above system of equations does not have any solution, for any finite value of $K$. Thus, the system is not $1-$sparse-controllable.
\end{exam}

Next example illustrates the benefits of using sparse control in a linear system over the sparse control with common support discussed in \Cref{sec:mmv_sparse}.

\begin{exam}
Consider a linear system with $N=3$, $L=3$,
\begin{equation}
\matD=\begin{bmatrix}
1&0&0\\
0&0&0\\
0&0&-1
\end{bmatrix} \text{, and } \matH=\begin{bmatrix}
0&1&0 \\
0&0&1\\
1&0&0
\end{bmatrix}.
\end{equation} 
This system satisfies the conditions in \Cref{thm:necessary_sufficient} for $s=2$, and is hence $2-$sparse-controllable. There are three possible unconstrained systems with input matrices of size $3\times 2$:
\begin{equation*}
\matH_{(1)}=\begin{bmatrix}
0&1 \\
0&0\\
1&0
\end{bmatrix} \hspace{0.5cm} \matH_{(2)}=\begin{bmatrix}
1&0 \\
0&1\\
0&0
\end{bmatrix} \hspace{0.5cm} \matH_{(3)}=\begin{bmatrix}
0&0 \\
0&1\\
1&0
\end{bmatrix}.
\end{equation*}
 However, the three subsystems described by the matrix pair $(\matD,\matH_{(k)})$ for $k=1,2,3$ are individually uncontrollable. Hence, sparse control allows the system to be controllable without adding much communication burden.
\end{exam} 

Finally, we give an example of a system with non-invertible $\matD$ which is both controllable and sparse-controllable. This example shows that the condition $\matD$ is invertible is not  necessary, but sufficient for a  controllable system to be sparse-controllable.

\begin{exam}
Consider a linear system with $N=3$, $L=2$,
\begin{equation}
\matD=\begin{bmatrix}
0&1&0\\
0&0&1\\
0&0&0
\end{bmatrix} \text{, and } \matH=\begin{bmatrix}
1  &   1 \\
1  &   0\\
1  &  1
\end{bmatrix}.
\end{equation} 
We note that $\matD$ is not an invertible matrix. Further,  the system satisfies the conditions in \Cref{thm:necessary_sufficient} for $s=1$, and hence it is $1-$sparse-controllable.
\end{exam}

To sum up, in this section, we answered  \ref{q:Q1} in \Cref{sec:intro}, and we address the question \ref{q:Q2}  in the next section.
}
\section{Minimum Number of Control Input Vectors} 
\label{sec:length}
In this section, we bound  the minimum number of input vectors that are required to drive the system from any given state to any final state. For comparison, we first state the corresponding result for the unconstrained system. In this section,  $q$ denotes the degree of the minimal polynomial of~$\matD$.
\begin{theorem}\label{thm:num_inp_uncon}
For a controllable system,  the minimum number of input vectors $K$ required to steer the system from any given state to any other state satisfies
\begin{equation}\label{eq:min_inp_uncon}
N/R_{\matH}\leq K \leq\min\lc q, N-R_{\matH}+1\rc\leq N.
\end{equation}
\end{theorem}
\begin{proof}
See \cite[Section 6.2.1]{Chen_Linear_1998}.
\end{proof}

We note that when we restrict the admissible inputs to sparse vectors, 
 the minimum number of input vectors  required can increase. This change is captured by the following theorem:
\begin{theorem}\label{thm:num_inp}
For an $s$-sparse-controllable system,  the minimum number of $s$-sparse input vectors $K^*$ required to steer the system from any given state to any other state satisfies
\begin{equation}\label{eq:min_inp}
\frac{N}{R_{\matH,s}^* }\!\leq\! K^*\!\leq\! \min \lc q\left\lceil \frac{ {S^*}}{s}\right\rceil, N-R_{\matH,s}^*+1 \rc \!\leq N,
\end{equation}
where 
{ $R_{\matH,s}^* \triangleq \min\lc R_{\matH},s\rc$ and 
\begin{align}
S^*&\triangleq \min \Big\{T: T=\lv\calS\rv \text{ for }\calS\subseteq [L] \notag\\
 &\hspace{0.8cm}\ld \text{ and }
\rank{\begin{bmatrix}
\matD-\lambda\eye & \matH_{\calS}\end{bmatrix}}=N, \forall \lambda\in\bbC\rc. \nonumber
\end{align}}
\end{theorem}
\begin{proof}
See \Cref{app:num_inp}. 
\end{proof}
{The above result can be intuitively explained as follows: At each time instant, we use at most $s$ linearly independent columns of $\matH$ to drive the system. Therefore, $R_{\matH}$ is replaced with $R_{\matH,s}^*$. Also, the first term  of the upper bound is computed by mapping the system to the reduced controllable system $(\matD,\matH_{\calS^*})$.  The reduced system retains the least number of columns of $\matH$ that are necessary to ensure controllability. Thus, under sparse inputs, we need $\left\lceil\lv\calS^*\rv/s\right\rceil$ times larger number of inputs compared to an unconstrained system.}
We make the following further observations from \Cref{thm:num_inp}.
\begin{itemize}
\item {Using the fact that that $S^*\leq R_{\matH}$, and from \Cref{thm:necessary_sufficient} which implies $R^*_{\matH,s} \geq \max\lc N-R_{\matD},1\rc$, we can get a relaxed bound instead of \eqref{eq:min_inp} as follows:
\begin{equation*}\label{eq:min_inp_relax}
\frac{N}{\min \lc R_{\matH},s\rc} \leq K^*\leq \min \lc q\left\lceil \frac{R_{\matH}}{s}\right\rceil, R_{\matD}+1, N \rc.
\end{equation*}}
\item 
The bound is invariant under right or left multiplication of $\matH$ by a non-singular matrix, and under any similarity transform on $\matD$. 
\item As $s$ increases, the system has more flexibility, and thus requires fewer number of input vectors to ensure controllability. Hence, the bounds are non-increasing in~$s$.
\item  {The upper and lower bounds in \Cref{thm:num_inp} meet when $N/R^*_{\matH,s}=N-R^*_{\matH,s}+1$, which gives $R^*_{\matH,s}$ as 1 or $N$. Similarly, for $s=1$, the lower and upper bounds in \Cref{thm:num_inp} are equal, and $K^*=N$. Further, if $R_{\matH}\geq s$, we get $R^*_{\matH,s}=s$, and thus the bounds are equal when $s=N$.}
\item We consider three cases for comparison with \Cref{thm:num_inp_uncon}:
\begin{enumerate}[label=(\alph*), leftmargin=0.3cm]
\item When $s=L$, which corresponds to the unconstrained case, \Cref{thm:num_inp} reduces to \Cref{thm:num_inp_uncon}, as expected. 

\item When {$s\geq S^*\geq R_{\matH}$}, \Cref{thm:num_inp} reduces to \Cref{thm:num_inp_uncon}, as $R^*_{\matH,s}=R_{\matH}$. This follows because when $s \ge R_{\matH}$, $\calCS\{\tilde{\matH}_{(K)}\}$ is the same as the column space of an $N\times Ks$ submatrix of $\tilde{\matH}_{(K)}$ with maximum rank. 

\item When $\min\lc q, N-\matR_{\matH}+1 \rc=N$, the system requires the same number of inputs to achieve controllability and $s$-sparse-controllability, for any $s$. However, this is possible only if $R_{\matH}=1$. When $s\geq R_{\matH}$, the system is equivalent to an unconstrained system as discussed above.
\end{enumerate} 

\end{itemize}
{
The following interesting corollary 
bounds the  number of $s$-sparse input vectors that ensures output controllability:
\begin{cor}\label{cor:num_inp_out}
For an output $s$-sparse-controllable system,  the minimum number of input vectors $K^*$ required to steer any initial output to any final output satisfies
\begin{multline}
\frac{m}{R_{\matA\matH,s}^*} \!\leq\!  K^* 
\!\leq\!  \min \lc q\left\lceil \frac{R_{\matH}}{s}\right\rceil, m-R_{\matA\matH,s}^*+1 \rc \!\leq\! m, \nonumber
\end{multline}
where $R_{\matA\matH,s}^*=\min\lc \rank{\matA\matH}, s\rc$.
\end{cor}

The bounds in \Cref{cor:num_inp_out} are smaller than those in \Cref{thm:num_inp}, because the dimension of the output space, $m$, is smaller than that of the state space, $N$. 
Further, substituting $s=L$ in \Cref{cor:num_inp_out}, we  see that for an output controllable system,  the minimum number of input vectors $K$ required to steer any initial output to any final output satisfies
\begin{equation*}
\frac{m}{ \rank{\matA\matH}} \leq  K 
\leq\min \lc q, m-\rank{\matA\matH}+1 \rc \leq m.
\end{equation*}
Similarly, we can extend \Cref{thm:num_inp} to the common support case discussed in \Cref{sec:mmv_sparse}:
\begin{cor}\label{cor:num_inp_stat}
For a system that is controllable using $s$-sparse inputs with a common support,  if $R_{\matH,s}^*=\min\lc R_{\matH},s\rc$, the minimum number of input vectors $K^*$ required to steer any initial output to any final output satisfies
\begin{equation}\label{eq:min_inp_stat}
\frac{N}{R_{\matH,s}^*} \leq  K^* 
\leq  \min \lc q, N-R_{\matH,s}^*+1 \rc \leq N.
\end{equation}
\end{cor}

\begin{proof}
The proof follows from \Cref{thm:num_inp_uncon} and the fact that there exists an index set $\calS\subseteq[L]$ such that $\lv\calS\rv=s$ and the system defined by $\lb\matD,\matH_{\calS}\rb$ is  controllable.
\end{proof}}

\section{Decomposing Sparse-controllable States}\label{sec:decomposition}
In this section, we consider \ref{q:Q3} in \Cref{sec:intro}, and present a  decomposition of the state space into sparse-controllable, sparse-uncontrollable and uncontrollable subspaces. We begin with the observation that $s$-sparse-controllability inherits  the \emph{invariance under a change of basis} property of the conventional controllability as discussed in the proposition below. 
\begin{prop}[Invariance under change of basis] The system defined by the matrix pair $(\matD, \matH)$ is $s$-sparse-controllable if and only if the system defined by $(\matU^{-1}\matD\matU, \matU^{-1}\matH)$ is $s$-sparse-controllable for every nonsingular $\matU\in\bbR^{N\times N}$.
\end{prop}

\begin{proof}
We note that when $\matD$ and $\matH$ are replaced with $\matU^{-1}\matD\matU$ and $\matU^{-1}\matH$ respectively, in \eqref{eq:control_mat_defn}, we get $\matU^{-1}\tilde{\matH}_{(K)}$ instead of $\tilde{\matH}_{(K)}$. Now,
the result follows from the Kalman-type rank test and the fact that the rank of every submatrix of $\tilde{\matH}_{(K)}$ and $\matU^{-1}\tilde{\matH}_{(K)}$ are the same.
\end{proof}

Inspired by the above proposition and in the same spirit as the Kalman decomposition~\cite{Kalman_Mathematical_1963}, we transform the original system to an equivalent \emph{standard form} using a change of basis, such that the   transformed state-space is separated into an $s$-sparse-controllable subspace and an orthogonal $s$-sparse-uncontrollable subspace. 
To this end, we first separate the controllable and uncontrollable states using the Kalman decomposition. Next, we identify the sparse-controllable part of the controllable part, for which we use the inequality condition of \Cref{thm:necessary_sufficient}. 
For this, we find a basis for the controllable part such that the transformed state-space separates into two subsystems: one which satisfies the inequality condition, and the other which does not. 
We now formally present the procedure for the decomposition, followed by an explanation of why the procedure works.
{
\begin{enumerate}[label=\arabic*.,leftmargin=0.3cm]
\item Find a basis for $\calCS\lc\tilde{\matH}_{(N)}\rc$ as $\lc\vecu_i\rc_{i=1}^{R}$, where $R\leq N$ is the rank of $\tilde{\matH}_{(N)}$. Extend the basis by adding $N-R$ linearly independent vectors $\lc\vecu_i\rc_{i=R+1}^N$ to define the invertible matrix 
$\matU\triangleq\begin{bmatrix}
\vecu_1 \ \vecu_{2}, \ldots ,\vecu_{N}
\end{bmatrix}\in\bbR^{N\times N}.$

\item Compute $\check{\matD} = \matU^{-1}\matD\matU\in\bbR^{N\times N}$ and $\check{\matH} = \matU^{-1}\matH\in\bbR^{N\times L}$ which take the following forms:
\begin{equation}
\check{\matD} = \begin{bmatrix}
\check{\matD}_{(1)} & \check{\matD}_{(2)}\\
\zero & \check{\matD}_{(3)}
\end{bmatrix} \hspace{0.5cm}
\check{\matH} = \begin{bmatrix}
\check{\matH}_{(1)} \\
\zero
\end{bmatrix},
\end{equation}
where $\check{\matD}_{(1)}\in\bbR^{R\times R}$ and $\check{\matH}_{(1)}\in\bbR^{R\times L}$.
\item Use the Jordan decomposition to get the following:
\begin{equation}
\check{\matD}_{(1)} = \matV \begin{bmatrix}
\bar{\matD}_{(11)}\in\bbR^{r\times r} & \zero\\
\zero & \zero
\end{bmatrix}\matV^{-1}
\end{equation}
 where  $ \matV\in\bbR^{R\times R}$ and $r\leq R$ is the rank of $\check{\matD}_{(1)}$.
\item  Define an invertible matrix $\matW\in\bbR^{N\times N}$ as follows:
\begin{equation}
\matW\triangleq \begin{bmatrix}
\matV\in\bbR^{R\times R} & \zero\in\bbR^{R\times N-R}\\
\zero \in \bbR^{N-R\times R} &
\eye \in \bbR^{N-R\times N-R}
\end{bmatrix}.
\end{equation}
\item Compute $\bar{\matD} = \matW^{-1}\check{\matD}\matW\in\bbR^{N\times N}$ and $\bar{\matH} = \matW^{-1}\check{\matH}\in\bbR^{N\times L}$, which take the following forms:
\begin{align}\notag
\bar{\matD} = \begin{bmatrix}
\bar{\matD}_{(11)}  & \zero&\bar{\matD}_{(21)}\in\bbR^{r\times N-R}\\
\zero  & \zero&\bar{\matD}_{(22)}\in\bbR^{R-r\times N-R}\\
\zero & \zero& \bar{\matD}_{(3)}\in\bbR^{N-R\times N-R}
\end{bmatrix} 
\bar{\matH} =\begin{bmatrix}
\bar{\matH}_{(1)} \\
\bar{\matH}_{(2)}\\
\zero
\end{bmatrix}.
\end{align}
where $\bar{\matH}_{(1)}\in\bbR^{r\times L} $ and $\bar{\matH}_{(2)}\in\bbR^{R-r\times L}$.
Define  $R_s\triangleq r+ \min\lc s,R-r\rc$. Then, the part of the state vector corresponding to the first $R_s$ entries is $s$-sparse-controllable, while the remaining part is $s$-sparse-uncontrollable. Also, since $\bar{\matD}=\lb\matU\matW\rb^{-1}\matD\lb\matU\matW\rb$ and $\bar{\matH}=\lb\matU\matW\rb^{-1}\matH$, the new basis is $\matU\matW$.\label{step:5}
\end{enumerate}

Here, steps 1 and 2 are the same as the Kalman decomposition, and in steps 3 and step 4, we find a basis that separates the sparse-controllable part from the controllable part. Let $\lb \matU\matW\rb^{-1}\vecx_k\tran =\begin{bmatrix}
\vecalpha_k\tran \in\bbR^{r}& \vecbeta_k\tran\in\bbR^{R-r} &\vecgamma_k\tran\in\bbR^{N-R}
\end{bmatrix}$. We then have the following equations which are equivalent to \eqref{eq:sys}:
\begin{align}
\vecalpha_k &= \bar{\matD}_{(11)}\vecalpha_{k-1}+\bar{\matD}_{(21)}\vecgamma_{k-1}+\bar{\matH}_{(1)}\vech_k\\
\vecbeta_k &= \bar{\matD}_{(22)}\vecgamma_{k-1}+\bar{\matH}_{(2)}\vech_k\\
\vecgamma_k &= \bar{\matD}_{(3)}\vecgamma_{k-1}.
\end{align}
Clearly, $\vecgamma_k$ is uncontrollable as it is independent of the input sequence. Further, the Kalman decomposition ensures that the part of the state vector corresponding to $\begin{bmatrix}
\vecalpha_k\tran &\vecbeta_k\tran
\end{bmatrix} \tran$ is controllable. Thus, 
\begin{equation}
\rank{\begin{bmatrix}
\bar{\matD}_{(11)}-\lambda\eye & \zero &\bar{\matH}_{(1)}\\
\zero & \zero &\bar{\matH}_{(2)}
\end{bmatrix}}=R,
\end{equation}
for any $\lambda\in\bbC$, and hence
\begin{equation}
\rank{\begin{bmatrix}
\bar{\matD}_{(11)}-\lambda\eye & \zero &\bar{\matH}_{(1)}\\
\zero & \zero &\lb\bar{\matH}_{(2)}\tran\rb_{\calS}\tran
\end{bmatrix}}=r+\lv\calS\rv,
\end{equation}
for any index set $\calS\subseteq[R-r]$. Therefore, from the inequality constraint of \Cref{thm:necessary_sufficient}, choosing $\lv\calS\rv=\min\lc s,R-r\rc$ ensures that the part of $\begin{bmatrix}
\vecalpha_k\tran & \vecbeta_{k,\calS}\tran
\end{bmatrix} \in \bbR^{R_s}$ corresponds to the sparse controllable part of the state vector. We choose $\calS$ as the top $R_s-r$ indices of the new state vector. We note that this holds because $\alpha_k$ is independent of $\beta_{k-1}$, and $\beta_k$ is independent of both $\alpha_{k-1}$ and $\beta_{k-1}$.
We illustrate the decomposition procedure with the following example.
\begin{exam}
Consider a linear system with $N=4$, $L=3$, $s=1$:
\begin{align*}
\!\!\!\matD \!=\!\! \begin{bmatrix}
	\!5.65  &       \!0 &  \!-1.25  & \!-7.95\\
    \!3.3   &     \!0   &\!-0.9 & \!-4.7\\
   \!-0.55  &     \! 0 & \! 0.35    &\! 0.85\\
    \!3.4   &      \!0  &\!-0.8   & \!-4.8\\
\end{bmatrix}
 \matH \!= \begin{bmatrix}
\!0.25  & \!1.25  & \! 1.5 \\
  \!  0.25 &   \!1.25  & \! 1.5\\
   \!-0.5  & \!-0.75 & \!-1.25\\
    \!0.25  & \! 1   & \!1.25\\
\end{bmatrix}
\end{align*}
Following the above procedure, from step 1
\begin{equation}
\matU = \begin{bmatrix}
1     & 0  &   4  &   1\\
     2  &  -1    & 3   &  0\\
    -2   &  0    &-1    & 1\\
     1    & 0     & 3   &  0
\end{bmatrix}.
\end{equation}
Step 2 gives the following with $R_1=3$:
\begin{equation}
\check{\matD}_{(1)} =
\begin{bmatrix}
0.2    &     0 &   0  \\
         0    &     0 &  0\\
0     &    0 &  0  \\
\end{bmatrix}, 
\check{\matH}_{(1)} = \begin{bmatrix}
0.25   &  0.25  &  0.5\\
    0.25    & 0  &  0.25\\
         0  &  0.25   & 0.25
\end{bmatrix}
\end{equation} 
Since $\check{\matD}_{(1)}$ is already in the Jordan form, $\matW=\eye$.
Finally, step 5 gives $R_s=2$
and the top $2$ entries of the state correspond to the $1-$sparse-controllable part of the system. It can be easily verified that the system defined using  \Cref{thm:necessary_sufficient}.
\end{exam}
\noindent \emph{Remark: } The linear system defined by $(\matD, \matH,\matA)$ is output $s$-sparse-controllable if and only if the system defined by $(\matU^{-1}\matD\matU, \matU^{-1}\matH, \matA\matU)$ is output $s$-sparse-controllable for any nonsingular  $\matU\in\bbR^{N\times N}$. 

 }

\section{Conclusions}
We presented two easily verifiable necessary and sufficient
conditions for controllability of linear
systems subject to sparsity constraints on the  input. 
Further, we bounded  the minimum number of sparse input vectors that ensure controllability.
The sparse-controllability tests led to a Kalman decomposition-like procedure for separating the system into sparse-controllable, controllable but sparse-uncontrollable, and uncontrollable parts. We also extended our results to the output controllability and controllability using sparse inputs with a common support.  However, our work does not impose any  constraint on the $\ell_{\infty}$ norm of the input vector, which may be required in applications where the maximum input magnitude is constrained. Addressing sparse-controllability under this constraint is an interesting avenue for future work.

\appendices
\crefalias{section}{app}

\section{Proof of \Cref{thm:necessary_sufficient}}
\label{app:necessary_sufficient}
\begin{proof}
We show that the conditions of the theorem are equivalent to the Kalman-type rank test. The proof relies on the fact that the Kalman rank test for the unconstrained system is equivalent to the PBH test, which is the same as the rank condition of \Cref{thm:necessary_sufficient}~\cite{Hautus_Stabilization_1970}.

We first prove that conditions of \Cref{thm:necessary_sufficient} imply the Kalman-type rank test. Suppose that the Kalman-type rank test fails. Then,  consider the following matrix of size $N\times N\tilde{K}s$:
\begin{multline}\label{eq:control_mat}
\tilde{\matH}^* =
[ \begin{matrix}
\matD^{\tilde{K} N-1\!}\matH_{\calS_1} &  \matD^{\!\tilde{K} N-2}\!\matH_{\calS_1} &\ldots 
\!\!& \matD^{(\tilde{K}-1\!) N}\!\matH_{\calS_1}  \end{matrix}\\
\begin{matrix}
\ldots & \matD^{(\tilde{K}-1) N-1}\matH_{\calS_2}&\ldots \matD^{(\tilde{K}-2) N}\matH_{\calS_2} &\ldots
\end{matrix}\\
\begin{matrix}
\ldots &
  \matD^{N-1}\matH_{\calS_{\tilde{K}}}&\ldots&\matH_{\calS_{\tilde{K}}}
\end{matrix}],
\end{multline}
where  we define $\tilde{K}\triangleq\lceil L/s\rceil$ index sets as follows:
\begin{equation}\label{eq:index_set}
\lv\calS_i\rv=s,\hspace{0.5cm} \cup_{i=1}^{\tilde{K}}\calS_i=[L]. 
\end{equation}
{We note that $\tilde{\matH}^*$ has the same form as that of the matrix for the Kalman-type rank test for sparse-controllability in \eqref{eq:submatrix_form}, with $K$ as $N\tilde{K}$.}
Since the Kalman-type rank test fails, $\tilde{\matH}^*$ does not have full row rank.  Further, we can rearrange the columns of $\tilde{\matH}^*$ to get the following matrix which has the same rank as that of $\tilde{\matH}^*$:
\begin{equation*}
\begin{bmatrix}
\matD^{N-1}\matH^* & \matD^{N-2}\matH^* & \ldots & \matH^*
\end{bmatrix},
\end{equation*} where we define the matrix $\matH^*\in\bbR^{N\times \tilde{K}s}$ as follows:
\begin{equation}\notag
\matH^*\triangleq\begin{bmatrix}
\matD^{(\tilde{K}-1)N}\matH_{\calS_1} & 
\matD^{(\tilde{K}-2)N}\matH_{\calS_2}
\ldots 
\matH_{\calS_{\tilde{K}}}
\end{bmatrix}.
\end{equation} Thus, using the classical Kalman rank test for the unconstrained inputs, the system defined by the matrix tuple $(\matD,\matH^*)$ is not controllable. Then, the classical PBH test for the unconstrained inputs implies that there exists $\lambda\in\bbC$ such that 
$\rank{\begin{bmatrix}
\matD-\lambda\eye \  \matH^*
\end{bmatrix}}<N.$
Therefore, there exists a nonzero vector $\vecz\in\bbR^N$ such that $\vecz\tran\matD=\lambda\vecz\tran$ and $\vecz\tran\matH^*=\zero$. However, we have,
\begin{equation}
\zero\! = \vecz\tran\matH^* \!= \vecz\tran\!\begin{bmatrix}
\lambda^{\!(\tilde{K}-1)N}\!\matH_{\calS_1} & 
\lambda^{\!(\tilde{K}-2)N}\!\matH_{\calS_2}
\ldots &
\matH_{\calS_{\tilde{K}}}
\!\end{bmatrix}.
\end{equation}
So either $\lambda=0$ and $\vecz\tran\matH_{\calS_{\tilde{K}}}=\zero$, or, if $\lambda\neq 0$,  $\vecz\tran\matH=\zero$ because $\vecz$ is orthogonal to all columns of $\matH$ due to \eqref{eq:index_set}. 
Hence, for every index set $\calS_i$ with $s$ entries, there exists $\vecz\in\bbR^N$ such that $\vecz\tran\matD=\lambda\vecz\tran$, and either
$\lambda=0$ and $\vecz\tran\matH_{\calS_{i}}=\zero$, or $\vecz\tran\matH=\zero$. Therefore, one of the following cases hold: 
\begin{enumerate}[label=(\alph*),leftmargin=0.3cm]
 \item There exists a left eigenvector $\vecz$ of $\matD$, such that $\vecz\tran\matH=\zero$. Thus, {the rank condition} of  \Cref{thm:necessary_sufficient} does not hold.
 \item {For every left eigenvector $\vecz$ of $\matD$, we have $\vecz\tran\matH\neq \zero$. However, for every index set $\calS$ with $s$ entries, there exists a nonzero vector $\vecz\in\bbR^N$ such that $\vecz\tran\matD=\zero$, and $\vecz\tran\matH_{\calS}=\zero$. This implies that $\rank{\begin{bmatrix}
 \matD & \matH
 \end{bmatrix} }=N$ and for every index set $\calS$ with $s$ entries, there exits $\vecz\in\bbR^N$ such that $\vecz\tran\begin{bmatrix}
 \matD & \matH_{\calS}
 \end{bmatrix}=\zero$. Therefore,  $s < N-R_{\matD}\leq R_{\matH}$. Thus, the inequality condition \Cref{thm:necessary_sufficient} does not hold.}
\end{enumerate}
Thus, when the Kalman-type rank test is unsuccessful, the conditions of the theorem are also violated.

Next, we prove that the Kalman-type rank test implies the conditions of the theorem. Suppose that the two conditions do not hold simultaneously. This could happen under the following two exhaustive cases:
\begin{enumerate}[label=(\alph*)]
\item Suppose that {the rank condition} does not hold. Then, the PBH test is violated which implies that the system is not controllable and thus, it cannot be sparse-controllable. 

\item Suppose that {the rank condition holds, but the inequality condition} does not hold. Then, for every index set $\calS$ with $s$ entries, there exists a nonzero vector $\vecz$ such that $\vecz\tran\matH_{\calS}=\zero$ and $\vecz\tran\matD=\zero$. This implies that for any set of $K>0$ index sets $\lc\calS_i:\lv\calS_i\rv=s\rc_{i=1}^K$ there exists a nonzero vector $\vecz\in\bbR^N$ such that 
\begin{equation}
\vecz\tran\begin{bmatrix}
\matD^{K-1}\matH_{\calS_1} & \matD^{K-2}\matH_{\calS_2} & \ldots & \matH_{\calS_K}
\end{bmatrix} = \zero.
\end{equation}
Hence, the Kalman-type rank test for fails.
\end{enumerate}
Thus, the proof is complete. 
\end{proof}

\section{Proof of \Cref{thm:num_inp}}\label{app:num_inp}
Using the Kalman-type rank test,  the minimum number of input vectors required to 
ensure controllability is the smallest integer $K$ that satisfies the rank condition of the test. So, for any  finite $K$, we define $\calH_{(K)}\subseteq\bbR^{N\times Ks}$ as the set of submatrices of $\tilde{\matH}_{(K)}$ of the form given in \eqref{eq:submatrix_form}.
Also, we define the following:
\begin{align}
R^*_{(K)} &= \underset{\matH_{(K)}\in\calH_{(K)}}{\max} \;\rank{\matH_{(K)}},\label{eq:R*_defn}\\
\calH^*_{(K)} &= \left\{\matH_{(K)}\in\calH_{(K)}: \rank{\matH_{(K)}}= R^*_{(K)}\right\} .
\end{align}
With these definitions, $K^*$ is the smallest integer such that $R^*_{(K^*)}=N$.

Before starting the proof, we outline the main steps involved. At a high level, there are five steps to the proof:
\begin{enumerate}[label=\Alph*.,leftmargin=0.3cm]
\item We begin by showing that for any matrix $\matH_{(K)}\in\calH_{(K)}$, we can find a matrix $\matH^*_{(K)}\in\calH^*_{(K)}$ such that
\begin{equation}
\calCS\lc\matH_{(K)}\rc \subseteq \calCS\lc\matH_{(K)}^*\rc.
\end{equation}
\item Second, using the above claim, we show that if $K$ is any integer such that 
\begin{equation}\label{eq:rank_con_index}
R^*_{(K)}=R^*_{(K+1)},
\end{equation} 
then
$R^*_{(K+Q)}=R^*_{(K)}$, for any positive integer $Q$. 
\item Third, we prove that $K^*$ is the smallest integer $K$ such that \eqref{eq:rank_con_index} holds,
which in turn leads to the upper bound: $K^*\leq  N+1-R_{\matH,s}^*$, where $R_{\matH,s}^*$ is as defined in the statement of the theorem.
\item Fourth, we show that in order to satisfy the rank criterion in \eqref{eq:rank_con_index}, $\matH^*_{(K^*)}$ needs to contain at most $qS^*$ number of columns with a particular structure. Then, we provide a choice of index sets $\lc\calS_i\rc_{i=1}^{K=q\lceil S^*/s\rceil}$ which can lead to that particular structure. Since the smallest integer $K$ that can achieve rank criterion in \eqref{eq:rank_con_index} is $K^*$, we assert that $K^*\leq q\lceil S^*/s\rceil$. Thus, together with the above step, we establish the upper bound in the theorem.
\item Finally, we lower bound $K^*$ to complete the proof.
\end{enumerate}
\subsection{Characterizing $\calH^*_{(K)}$}
If $\matH_{(K)}\in\calH^*_{(K)}$, the result is trivial: $\matH_{(K)}^*=\matH_{(K)}$. Suppose that $\matH_{(K)}\notin\calH^*_{(K)}$, then 
$\rank{\matH_{(K)}}< R^*_{(K)}$. Therefore, to find $\matH_{(K)}^*$, we have to replace some linearly dependent columns of $\matH_{(K)}$ with columns which are linearly independent of the rest of the columns of $\matH_{(K)}$, as follows:
\begin{enumerate}[label=(\alph*),leftmargin=0.3cm]
\item Find a set $\lc\vecu_i\rc_{i=1}^{\rank{\matH_{(K)}}}$ of columns of $\matH_{(K)}$ that are  linearly independent and span $\calCS\lc\matH_{(K)}\rc$. 

\item Since $\matH_{(K)}$ is a submatrix of $\tilde{\matH}_{(K)}$, we can extend the set $\lc\vecu_i\rc_{i=1}^{\rank{\matH_{(K)}}}$ to form a basis $\lc\vecu_i\rc_{i=1}^{\rank{\tilde{\matH}_{(K)}}}$ of $\calCS\lc\tilde{\matH}_{(K)}\rc$ by adding  columns from $\tilde{\matH}_{(K)}$. We note that $\vecu_i=\matD^p\matH_j$ for some integers $p$ and $j$ because of the structure of $\tilde{\matH}_{(K)}$.

\item  Replace the linearly dependent columns of $\matH_{(K)}$ with the columns from the set $\lc\vecu_i\rc_{i=\rank{\matH_{(K)}}+1}^{\rank{\tilde{\matH}_{(K)}}}$ to get a new matrix $\bar{\matH}_{(K)}\in\bbR^{N\times Ks}$. 
We only replace a column of the form $\matD^p\matH_j$ in $\matH_{(K)}$ with another column of the form $\matD^p\matH_{j'}$, for all $p$ and $j$ and some integer $j'$. This ensures that $\bar{\matH}_{(K)}\in\calH_{(K)}$. In this fashion, we replace as many columns of $\matH_{(K)}$ as necessary to ensure that $\bar{\matH}_{(K)}$ has the maximum rank, $R^*_{(K)}$. However, since we are only replacing linearly dependent columns, we have
\begin{equation}\label{eq:CS_replace_1}
\calCS\lc\matH_{(K)}\rc \subseteq \calCS\lc\bar{\matH}_{(K)}\rc.
\end{equation}
\end{enumerate}
Since $\rank{\bar{\matH}_{(K)}}=R^*_{(K)}$ and $\bar{\matH}_{(K)}\in\calH_{(K)}$, we get that $\bar{\matH}_{(K)}\in\calH_{(K)}^*$, satisfying \eqref{eq:CS_replace_1}. Hence, the first step of the proof is complete.
\subsection{Characterizing $R^*_{(K)}$}
We use induction to show that $R^*_{(K+Q)}=R^*_{(K)}$, for any integer $Q>0$. Hence, it suffices to show the following:
\begin{equation}\label{eq:goal_step_1}
R^*_{(K+2)}=R^*_{(K+1)}.
\end{equation}
From \eqref{eq:R*_defn}, we know that $R^*_{(K+2)}\geq R^*_{(K+1)}$. Also, 
\begin{equation}
R^*_{(K)} = \underset{\matH_{(K)}\in\calH_{(K)}}{\max}\text{dim}\lc\calCS\lc \matH_{(K)}\rc\rc,
\end{equation}
where $\text{dim}\{\cdot\}$ denotes the dimension of a subspace. Thus,
we establish \eqref{eq:goal_step_1} by showing that for any matrix $\matH_{(K+2)}\in\calH_{(K+2)}$, there exists a matrix $\matH_{(K+1)}^*\in\calH^*_{(K+1)}$ such that
\begin{equation}
\calCS\lc \matH_{(K+2)}\rc\subseteq \calCS\lc\matH_{(K+1)}^*\rc.
\end{equation}

We prove this relation by separately looking at the column spaces  spanned by the first $s$ columns and the last $(K+1)s$ columns of $\matH_{(K+2)}$. We know that the submatrix formed by the last $(K+1)s$ columns of any matrix in $\calH_{(K+2)}$ belongs to $\calH_{(K+1)}$. Thus, using the claim in the first step, we can find a matrix $\matH^*_{(K+1)}$ such that the column space spanned by the last $(K+1)s$ columns of $\matH_{(K)}$ is contained in $\calCS\lc\matH^*_{(K+1)}\rc$. Therefore, it suffices to show that the column space spanned by the first $s$ columns of $\matH_{(K+2)}$ is contained in the column space of 
$\matH^*_{(K+1)}$. 

To prove the above statement, we note that the column space of the first $s$ columns of $\matH_{(K+2)}$ is contained in $\calCS\lc\matD^{K+1}\matH\rc$. Also, $\calCS\lc\matH^*_{(K+1)}\rc$ 
contains $\underset{\matH^*_{(K+1)}\in\calH^*_{(K+1)}}{\cap}\calCS\lc\matH^*_{(K+1)}\rc$. Hence, it suffices to show that
\begin{equation}\label{eq:goal_step_2}
\calCS\lc\matD^{K+1}\matH\rc \subseteq \underset{\matH^*_{(K+1)}\in\calH^*_{(K+1)}}{\cap}\calCS\lc\matH^*_{(K+1)}\rc,
\end{equation}
which we prove using the relation \eqref{eq:rank_con_index}. 

To show that \eqref{eq:goal_step_2} holds, we consider an index set $\calS\subseteq[L]$ with $s$ entries and a matrix $\matH^*_{(K)}\in\calH^*_{(K)}$. Now, the matrix $\begin{bmatrix}
\matD^{K}\matH_{\calS} & \matH^*_{(K)}
\end{bmatrix}\in\bbR^{N\times (K+1)s}$ belongs to $\calH_{(K+1)}$. Thus, from \eqref{eq:R*_defn} and \eqref{eq:rank_con_index} we have 
\begin{equation}
\rank{\begin{bmatrix}
\matD^{K}\matH_{\calS} & \matH^*_{(K)}
\end{bmatrix}} \leq R^*_{(K+1)}=R^*_{(K)}.
\end{equation}
However, we also have
\begin{equation}
\rank{\begin{bmatrix}
\matD^{K}\matH_{\calS} & \matH^*_{(K)}
\end{bmatrix}}\geq \rank{ \matH^*_{(K)}} \\
= R^*_{(K)}.
\end{equation}
Thus,  for all index sets $\calS$ with $s$ entries and any matrix $\matH^*_{(K)}\in\calH^*_{(K)}$,
\begin{equation}
\rank{\begin{bmatrix}
\matD^{K}\matH_{\calS} & \matH^*_{(K)}
\end{bmatrix}} = \rank{ \matH^*_{(K)}}
\end{equation}
 This relation immediately implies the following:
\begin{equation}
\rank{\begin{bmatrix}
\matD^{K}\matH & \matH^*_{(K)}
\end{bmatrix}} = \rank{\matH^*_{(K)}},\label{eq:rank_submatrix}
\end{equation}
for any matrix $\matH^*_{(K)}\in\calH^*_{(K)}$.
Thus, we get that the columns of $\matD^K\matH$ belong to $\calCS\lc\matH^*_{(K)}\rc$, for any matrix $\matH^*_{(K)}\in\calH^*_{(K)}$. Hence, 
\begin{equation}
\calCS\lc\matD^{K}\matH\rc \subseteq \underset{\matH^*_{(K)}\in\calH^*_{(K)}}{\cap}\calCS\lc\matH^*_{(K)}\rc.
\end{equation} 
Therefore, we get
\begin{equation}\label{eq:CS_DH_subset}
\calCS\lc\matD^{K+1}\matH\rc \subseteq \underset{\matH^*_{(K)}\in\calH^*_{(K)}}{\cap}\calCS\lc\matD\matH^*_{(K)}\rc.
\end{equation} 
Hence, to prove \eqref{eq:goal_step_2}, we need to show that 
\begin{equation}\label{eq:goal_step_3}
\underset{\matH^*_{(K)}\in\calH^*_{(K)}}{\cap}\calCS\lc\matD\matH^*_{(K)}\rc \subseteq \underset{\matH^*_{(K+1)}\in\calH^*_{(K+1)}}{\cap}\calCS\lc\matH^*_{(K+1)}\rc.
\end{equation}

We prove the above relation by showing that
there exists a matrix $\matH^*_{(K+1)}\in \calH^*_{(K+1)}$ such that 
\begin{equation}\label{eq:goal_step_4}
\calCS\lc\matD\matH^*_{(K)}\rc \subseteq \calCS\lc\matH^*_{(K+1)}\rc,
\end{equation}
for every matrix $\matH^*_{(K)}\in\calH^*_{(K)}$.
So we consider a new matrix $\bar{\matH}_{(K+1)} \in \bbR^{N\times (K+1)s}$ as follows:
\begin{equation}\label{eq:barH_defn}
\bar{\matH}_{(K+1)} \triangleq \begin{bmatrix}
\matD\matH^*_{(K)} & \matH_{\calS}
\end{bmatrix},
\end{equation}
for some index set $\calS\subseteq[L]$ and $\lv\calS\rv=s$. Since $\bar{\matH}_{(K+1)}\in\calH_{(K+1)}$, using the arguments in the first step, we can find a matrix $\matH^*_{(K+1)}\in\calH^*_{(K+1)}$ such that  
\begin{equation}
\calCS\lc\bar{\matH}_{(K+1)} \rc \subseteq \calCS\lc\matH^*_{(K+1)}\rc.
\end{equation}
However, \eqref{eq:barH_defn} implies that
$\calCS\lc\matD\matH^*_{(K)} \rc \subseteq \calCS\lc\bar{\matH}_{(K+1)}\rc.$ 
Therefore, \eqref{eq:goal_step_4} holds, and hence \eqref{eq:goal_step_3} is proved. 

Recall that \eqref{eq:goal_step_3} implies \eqref{eq:goal_step_2}, which in turn establishes the relation \eqref{eq:goal_step_1}. By mathematical induction, we conclude that $\rank{\matH^*_{(K+Q)}}=\rank{\matH^*_{(K)}}$, for any positive integer $Q$, and the proof of the second step in the outline is complete.

\subsection{First part of the upper bound}
Suppose that $K_*$ is the smallest integer such that $R^*_{(K_*)}=R^*_{(K_*+1)}$. 
From \eqref{eq:R*_defn}, it is clear that 
\begin{equation}
R^*_{(K)}\leq R^*_{(K+1)}\leq N,
\end{equation} for any positive integer $K$. Since $R^*_{(K^*)}=N$, we have $R^*_{(K^*)}=R^*_{(K^*+1)}=N$. Therefore, $K_*\leq K^*$, and $R^*_{(K_*)}=N$ from the claim in the second step.

Further, since $K^*$ is the smallest integer such that $R^*_{(K^*)}=N$, we have  $K_*=K^*$. Hence,  $R^*_{(K)}$ strictly increases with $K$, for $1\leq K \leq K^*$, and we have 
\begin{align}
N &= R^*_{(K^*)} \geq R^*_{(K^*-1)}+1 \geq R^*_{(K^*-2)}+2\notag \\
&\geq R^*_{(1)} + K^*-1= R_{\matH,s}^*+ K^*-1.
\end{align}
Hence, the third step in the outline is complete.

\subsection{Upper bounding $K^*$}
To prove that $K^*\leq q\lceil {S^*}/s\rceil$,  we  first look at the linearly independent columns in $\matH^*_{(K^*)}$. For any $K$, each column of $\matH^*_{(K)}$ is of the form $\matD^p\matH_j$, for some integer $p$, and $j\in[L]$. However, since $q$ is the degree of the minimal polynomial of $\matD$,  for any integer $Q$, $\matD^p$ can be expressed as a linear combination of $\lc\matD^{i}\rc_{i=Q}^{Q+q-1}$, for all $p\geq Q$. Therefore,
for any $j$, if $\lc\matD^{i}\matH_j\in\bbR^N\rc_{i=Q}^{Q+q-1}$ are any $q$ columns of $\matH^*_{(K)}$, further adding columns of the form $\matD^p\matH_j$,  $p\geq Q$, does not improve the rank of the matrix. Therefore, for a given $j$, at most $q$ columns of the form $\matD^p\matH_j$ need to be present in $\matH^*_{(K)}$ to ensure the rank criterion in \eqref{eq:rank_con_index}. 

{Further, let  $\matH_{\calS'}$ with $\calS'\subseteq[L]$ represent the smallest set of columns of $\matH$ such that the linear system described by the tuple $\lb\matD,\matH_{\calS}\rb$ is controllable. As given in the statement of the theorem, let $S^*=\lv\calS'\rv$.} Then, for any integer $p$, if $\lc\matD^{p}\matH_j\in\bbR^N\rc_{j\in\calS'}$ are any {$S^*$} columns of $\matH^*_{(K)}$, further adding columns of the form $\matD^p\matH_j$, for $j\notin \calS'$ does not improve the rank of the matrix. Thus, for any given $p$, at most {$S^*$} columns of the form $\matD^p\matH_j$ need to be present in $\matH^*_{(K)}$ to ensure the rank criterion. 

In short, we have proved that, in order to ensure the rank criterion in \eqref{eq:rank_con_index}, $\matH^*_{(K)}$ needs to have at most $q$ columns of the form $\matD^p\matH_j$, for any given $j$, and  at most {$S^*$} columns of the form $\matD^p\matH_j$, for any given $p$. Hence, $\matH^*_{(K)}$ needs to have  at most $q{S^*}$ columns to satisfy the rank criterion in~\eqref{eq:rank_con_index}.


Finally, we provide a choice of index sets for each input vector, that satisfies the above conditions. We form  index sets $\lc\calS_i'\rc_{i=1}^{K=\lceil {S^*}/s\rceil}$ that partition the set of  ${S^*}$  columns into groups of size at most $s$. The index sets are selected such that $\cup_{i=1}^{K}\calS'_i=\calS'$, $\lv\calS_i\rv=s$, and $\calS_{K}$ is such that 
 $\begin{bmatrix}
\matD & \matH_{\calS_K}
\end{bmatrix}$ has rank $N$. The existence of such an index set $\calS_K$ is ensured by \Cref{thm:necessary_sufficient}, and they need not be disjoint. Next, we choose $\calS_i=\calS_j'$, for $i=(j-1)q+1,(j-1)q+2,\ldots,jq$. Hence, we get the following  submatrix   of $\tilde{\matH}_{(K)}\in\bbR^{N\times qKL}$:
\begin{multline}
\matH^*_{(K)} = [ \begin{matrix}
\matD^{K q-1\!}\matH_{\calS_1} &  \matD^{\!K q-2}\!\matH_{\calS_1} &\ldots 
\!\!& \matD^{(K-1\!) q}\!\matH_{\calS_1}  \end{matrix}\\
\begin{matrix}
\ldots & \matD^{(K-1) q-1}\matH_{\calS_2}&\ldots \matD^{(K-2) q}\matH_{\calS_2} &\ldots
\end{matrix}\\
\begin{matrix}
\ldots &
  \matD^{q-1}\matH_{\calS_{K}}&\ldots&\matH_{\calS_{K}}
\end{matrix}].
\end{multline}
It is easy to see that this choice of index sets ensures that for any given $p$, ${S^*}$ columns of the form $\matD^p\matH_j$ are present in $\matH^*_{(K)}$. Also, for any given $j\in\calS'$, $q$ columns of $\lc\matD^{i}\matH_j\in\bbR^N\rc_{i=Q}^{Q+q-1}$ are present in 
$\matH^*_{(K)}$. Hence, $K^*\leq q\lceil {S^*}/s\rceil$, which establishes the upper bound in \eqref{eq:min_inp}.

\subsection{Lower bounding $K^*$}
The lower bound is achieved when all columns of $\matH_{(K)}^*$ are linearly independent. Thus, to ensure that rank $\matH_{(K)}^*$ is $N$, $Ks\geq N$. However, if $s\geq R_{\matH}$, the maximum number of independent columns become $KR_{\matH}$, and thus we get that $KR_{\matH}\geq N$. Hence, $K\min\lc R_{\matH},s\rc\leq N$, and the lower bound in \eqref{eq:min_inp} is proved.

As noted in the proof outline, this suffices to establish \Cref{thm:num_inp}.\bibliographystyle{IEEEtran}
\bibliography
{IEEEabrv,bibJournalList,SparseControllability}
\par\leavevmode

\begin{IEEEbiography}
[{\includegraphics
[width=1in,height=1.25in,clip,keepaspectratio]
{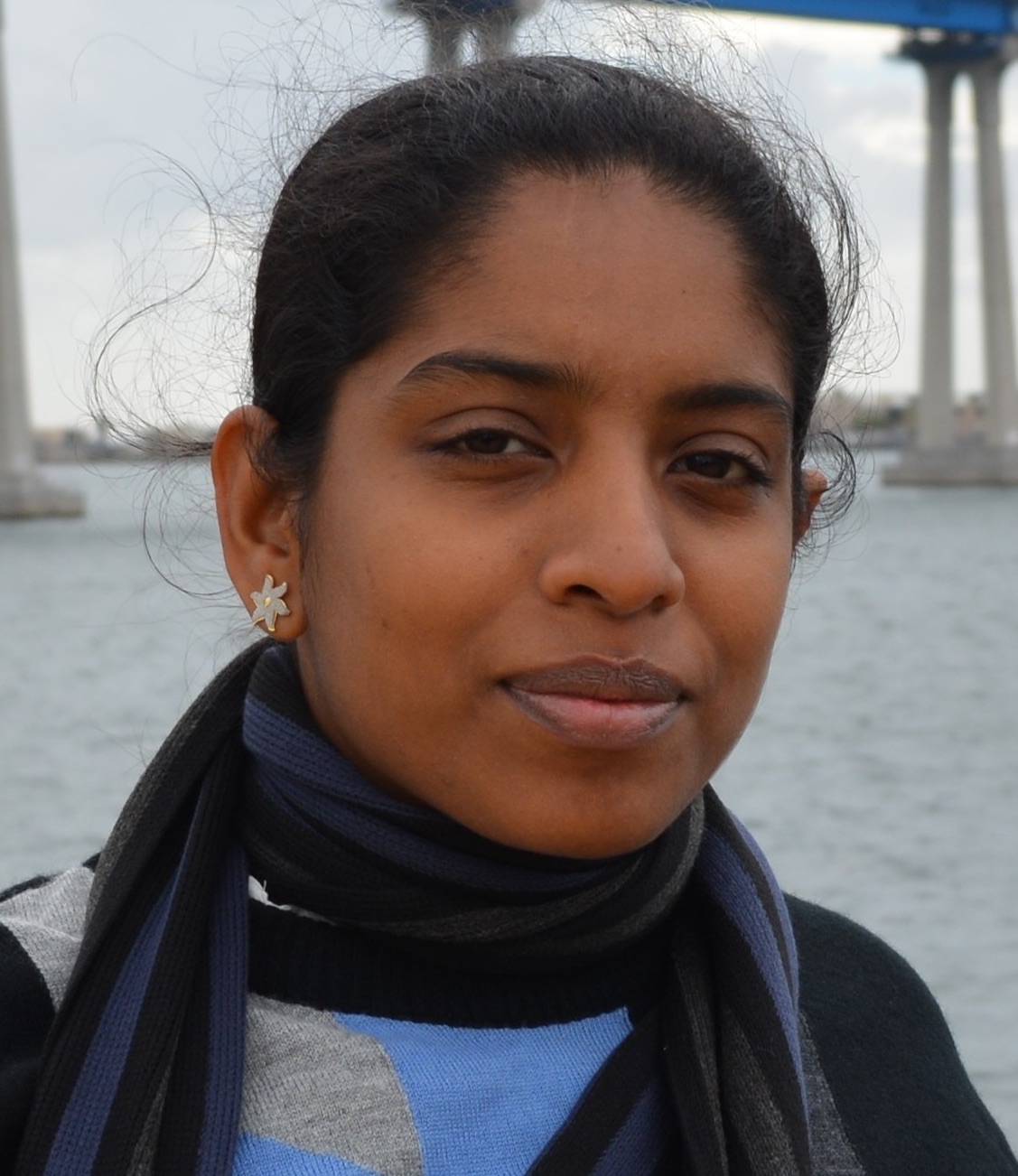}}]
{Geethu Joseph}
received the B.~Tech.~degree in Electronics and Communication Engineering from National Institute of Technology, Calicut, India, in
2011, and the M.\ E.\ degree in Signal Processing and the Ph.D. degree in Electrical Communication Engineering (ECE), from the Indian Institute of Science (IISc), Bangalore, in 2014 and 2019, respectively. She was awarded the Prof. I. S. N. Murthy medal in 2014 for being the best M.~E.~(signal processing) student in the Dept.\ of ECE, IISc. She is currently a post-doctoral fellow at the Department of Electrical Engineering and Computer Science, Syracuse University, NY. Her research interests include statistical signal processing, adaptive filter theory, sparse Bayesian learning, and compressive sensing.
\end{IEEEbiography}

\begin{IEEEbiography}
[{\includegraphics
[width=1in,height=1.25in,clip,keepaspectratio]
{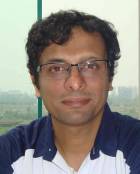}}]
{Chandra R. Murthy}
received the B.\ Tech.\ degree in Electrical Engineering from the Indian Institute of Technology Madras,
Chennai, India, in 1998, the M.S. and Ph.D. degrees
in Electrical and Computer Engineering from Purdue
University, West Lafayette, IN and the University
of California, San Diego, CA, in 2000 and 2006,
respectively. From 2000 to 2002, he worked as an engineer for
Qualcomm Inc., San Jose, USA, where he worked on
WCDMA baseband transceiver design and 802.11b
baseband receivers. From 2006 to 2007, he worked as a staff
engineer at Beceem Communications Inc., Bangalore, India on advanced receiver architectures for the 802.16e Mobile WiMAX standard. Currently, he is working as a Professor in the Department of Electrical Communication Engineering at the Indian Institute of Science, Bangalore, India.

His research interests are in the areas of energy harvesting communications, multiuser MIMO systems, and sparse signal recovery techniques applied to wireless communications. His paper won the best paper award in the Communications Track at NCC 2014 and a paper co-authored with his student won the student best paper award at the IEEE ICASSP 2018. He has 50+ journal papers and 80+ conference papers to his credit. He was an associate editor for the IEEE Signal Processing Letters during 2012-16. He is an elected member of the IEEE SPCOM Technical Committee for the years 2014-16, and has been re-elected for the 2017-19 term. He is a past Chair of the IEEE Signal Processing Society, Bangalore Chapter. He is currently serving as an associate editor for the IEEE Transactions on Signal Processing and IEEE Transactions on Information Theory and as an editor for the IEEE Transactions on Communications.
\end{IEEEbiography}

\end{document}